\documentclass[letterpaper,10pt,conference]{ieeeconf}  
\IEEEoverridecommandlockouts                              
\pdfoutput=1

\usepackage{xspace}
\usepackage[english]{babel}
\usepackage[latin1]{inputenc}

\usepackage{amsmath,amssymb}
\usepackage{enumerate}

\usepackage{tikz}
\usetikzlibrary{automata,positioning}
\tikzset{initial text={},
    every state/.style={circle,minimum size=.3cm,draw=blue!50,very thick,fill=blue!20},
    secret/.style={minimum size=.4cm,draw=red!50,very thick,fill=red!20,rectangle},
    node distance=1.5cm,on grid,auto,
    bend angle=65}
\usetikzlibrary{positioning,calc,chains}
\usetikzlibrary{automata,fit}
\usetikzlibrary{backgrounds}
\usetikzlibrary{decorations.pathreplacing}
\usetikzlibrary{fadings}
\usetikzlibrary{shapes.geometric}
\usetikzlibrary{decorations.pathreplacing}

\def\ie{{i.e.},~}
\def\eg{{e.g.},~}

\def\st{{s.t.}~}
\def\emptyset{\varnothing}

\def\Trace{\textit{Tr}}
\def\trace{\textit{tr}}

\def\faulty{\textit{Faulty}}

\def\nonfaulty{\textit{NonFaulty}}

\def\runs{\textit{Runs}} \def\lang{{\cal L}}  \def\cost{\textit{Cost}}
\newcommand{\proj}[1]{\boldsymbol{\pi}_{/#1}} 

\newcommand{\vect}[1]{\mathbf{#1}}

\newtheorem{prob}{Problem}  
\newtheorem{definition}{Definition} 
\newtheorem{theorem}{Theorem} 
\newtheorem{remark}{Remark} 
\newtheorem{example}{Example}

\newcommand{\sem}[1]{[\![#1]\!]}

\newcommand{\setN}{\mathbb N}
\newcommand{\setR}{\mathbb R}
\newcommand{\setB}{\mathbb B}
\newcommand{\setZ}{\mathbb Z}
\newcommand{\setQ}{\mathbb Q}

\def\calA{{\cal A}}

\def\calC{{\cal C}}

\def\calU{{\cal U}}
\def\cc{\calC}

\def\last{\textit{last}}

\def\Cost{\textit{Cost}}
\def\MeanCost{\textit{$\overline{\Cost}$}}

\def\endef{\ifmmode\squareforged\else{\unskip\nobreak\hfil
\penalty50\hskip1em\null\nobreak\hfil$\blacksquare$
\parfillskip=0pt\finalhyphendemerits=0\endgraf}\fi}

\def\endex{\ifmmode\squareforged\else{\unskip\nobreak\hfil
\penalty50\hskip1em\null\nobreak\hfil$\square$
\parfillskip=0pt\finalhyphendemerits=0\endgraf}\fi}

\def\ssi{iff\xspace}

\def\motvide{\varepsilon}
\def\tauac{\tau}

\newcommand{\dur}{{\textit{Dur}}} 
\def\inv{\textit{Inv}}

\def\tw{\textit{TW\/}}

\def\untimed{\textit{Unt}} \def\dta{DTA\xspace}
\def\dtamu{DTA$_\mu$\xspace}

\def\true{\mbox{\textsc{true}}}
\def\false{\mbox{\textsc{false}}}
\def\obs{\textit{Obs}\xspace}
\def\rg{\textit{RG}}

\title{\LARGE \bf Dynamic Observers for Fault Diagnosis of Timed
  Systems}

\author{Franck Cassez, \IEEEmembership{Member, IEEE}%
  \thanks{Franck Cassez is with National ICT Australia \& CNRS, Locked
    Bag 6016, The University of New South Wales, Sydney NSW~1466,
    Australia. \texttt{\scriptsize
      franck.cassez@cnrs.irccyn.fr,Franck.Cassez@nicta.com.au}}%
  \thanks{Author supported by a Marie Curie International Outgoing
    Fellowship within the 7th European Community Framework Programme.}
}

\begin{document}
\maketitle
  
\thispagestyle{empty}

\begin{abstract} 
  In this paper we extend the work on \emph{dynamic ob\-servers} for
  fault diagnosis~\cite{cassez-acsd-07,cassez-tase-07,cassez-fi-08} to
  timed automata. We study sensor minimization problems with static
  observers and then address the problem of computing the most
  permissive dynamic observer for a system given by a timed automaton.
\end{abstract}

\section{Introduction}

Discrete-event systems~\cite{RW87} (DES) can be modelled by finite
automata over an alphabet of actions/events $\Sigma$.  The fault
diagnosis problem~\cite{Raja95} for DES consists in detecting
\emph{faulty} sequences in the system.
A \emph{faulty} sequence is a sequence of the DES containing an
occurrence of a special event $f$.  It is assumed that an external
\emph{observer} which has to detect faults, knows the
specification/model of the DES, but can partially observe the system
at runtime: it is able to observe sequences of \emph{observable}
events in $\Sigma_o \subseteq \Sigma$.  Based on this knowledge, it
has to announce whether an observation (in $\Sigma_o^*$) stems from a
faulty sequence (in $(\Sigma \cup \{\tauac,f\})^*$). 
Checking diagnosability of DES can be done in PTIME and computing a
diagnoser amounts to determinizing the DES
(EXPTIME)~\cite{Raja95,Jiang-01,yoo-lafortune-tac-02}.
\smallskip

\noindent{\it \bfseries Fault Diagnosis for Timed Automata.} 
The fault diagnosis problem for Timed Automata (TA) has been
introduced and solved by S.~Tripakis in~\cite{tripakis-02}, where he
proved that checking {diagnosability} of a timed automaton is
PSPACE-complete.  In the timed case however, the diagnoser may be a
Turing machine.  In a subsequent work by P.~Bouyer~and
F.~Chevalier~\cite{Bouyerfossacs05}, the problem of checking whether a
timed automaton is diagnosable using a diagnoser which is a
\emph{deterministic} timed automaton (DTA) was studied, and they
proved that this problem was 2EXPTIME-complete.
\smallskip

\noindent{\it \bfseries Our Contribution and Related Work.} 
In~\cite{cassez-acsd-07,cassez-tase-07} (and~\cite{cassez-fi-08} for
an extended version), we have introduced \emph{dynamic observers} for
fault diagnosis of DES.  In this framework, an observer can choose
dynamically which events it is going to observe and make a new choice
after each occurrence of any (currently) observable event.
In~\cite{cassez-acsd-07,cassez-fi-08} we have shown how to compute
(2EXPTIME) a \emph{most permissive observer} which represents all the
the dynamic observers that ensures that a DES is diagnosable.
In~\cite{cassez-tase-07} we have furthermore introduced a notion of
\emph{cost} of an observer, and proved that an optimal observer could
also be computed in 2EXPTIME.

In this paper, we extend the previous results for systems given by
timed automata. We first settle the complexity of some optimization
problems with static observers (section~\ref{sec-static}).  We then
focus on dynamic \emph{timed} observers, and show how to compute
(section~\ref{sec-dynamic}) a most permissive (timed) dynamic
observer, under the assumption of bounded \emph{resources}. In
section~\ref{sec-cost}, we define a notion of \emph{cost} for timed
observers (which extends the one we have defined for DES
in\cite{cassez-tase-07}) and show how to compute the cost of a given
observer.  We also discuss the problem of synthesizing an optimal
timed dynamic observer.


\section{Preliminaries}\label{sec-prelim}
$\Sigma$ denotes a finite alphabet and $\Sigma_\tauac=\Sigma \cup
\{\tauac\}$ where $\tauac \not\in \Sigma$ is the \emph{unobservable}
action.  $\setB=\{\true,\false\}$ is the set of boolean values,
$\setN$ the set of natural numbers, $\setZ$ the set of integers and
$\setQ$ the set of rational numbers.  $\setR$ is the set of real
numbers and $\setR_{\geq 0}$ is the non-negative real numbers.

\subsection{Clock Constraints}
Let $X$ be a finite set of variables called \emph{clocks}.  A
\emph{clock valuation} is a mapping $v : X \rightarrow \setR_{\geq
  0}$. We let $\setR_{\geq 0}^X$ be the set of clock valuations over
$X$. We let $\vect{0}_X$ be the \emph{zero} valuation where all the
clocks in $X$ are set to $0$ (we use $\vect{0}$ when $X$ is clear from
the context).  Given $\delta \in \setR$, $v + \delta$ denotes the
valuation defined by $(v + \delta)(x)=v(x) + \delta$. We let $\cc(X)$
be the set of \emph{convex constraints} on $X$, \ie the set of
conjunctions of constraints of the form $x \bowtie c$ with $c
\in\setZ$ and $\bowtie \in \{\leq,<,=,>,\geq\}$. Given a constraint $g
\in \cc(X)$ and a valuation $v$, we write $v \models g$ if $g$ is
satisfied by $v$.  Given $R \subseteq X$ and a valuation $v$, $v[R]$
is the valuation defined by $v[R](x)=v(x)$ if $x \not\in R$ and
$v[R](x)=0$ otherwise.

\subsection{Timed Words}
The set of finite (resp. infinite) words over $\Sigma$ is $\Sigma^*$
(resp. $\Sigma^\omega$) and we let $\Sigma^\infty=\Sigma^* \cup \Sigma
^\omega$. We let $\varepsilon$ be the empty word. A \emph{language}
$L$ is any subset of $\Sigma^\infty$. A finite (resp. infinite)
\emph{timed word} over $\Sigma$ is a word in $(\setR_{\geq
  0}.\Sigma)^*.\setR_{\geq 0}$ (resp. $(\setR_{\geq
  0}.\Sigma)^\omega$).  $\dur(w)$ is the duration of a timed word $w$
which is defined to be the sum of the durations (in $\setR_{\geq 0}$)
which appear in $w$; if this sum is infinite, the duration is
$\infty$.  Note that the duration of an infinite word can be finite,
and such words which contain an infinite number of letters, are called
\emph{Zeno} words.

$\tw^*(\Sigma)$ is the set of finite timed words over $\Sigma$,
$\tw^\omega(\Sigma)$, the set of infinite timed words and
$\tw(\Sigma)=\tw^*(\Sigma) \cup \tw^\omega(\Sigma)$. A \emph{timed
  language} is any subset of $\tw(\Sigma)$. 

In this paper we write timed words as $0.4\ a\ 1.0\ b\ 2.7 \ c \cdots$
where the real values are the durations elapsed between two letters:
thus $c$ occurs at global time $4.1$. 
We let $\untimed(w)$ be the \emph{untimed} ver\-sion of $w$ ob\-tai\-ned by
erasing all the durations in $w$,
\eg $\untimed(0.4\ a\ 1.0\ b\ 2.7 \
c)=abc$.
 Given a timed language $L$, we
let $\untimed(L)=\{ \untimed(w) \ | \ w \in L \}$.

Let $\proj{\Sigma'}$ be the projection of timed words of $\tw(\Sigma)$
over timed words of $\tw(\Sigma')$.  When projecting a timed word $w$
on a sub-alphabet $\Sigma' \subseteq \Sigma$, the durations elap\-sed
bet\-ween two events are set accordingly: for instance
$\proj{\{a,c\}}(0.4 \ a\ 1.0\ b\ 2.7 \ c )=0.4 \ a \ 3.7 \ c$
(projection erases some letters but keep the time elapsed between two
letters).    Given $\Sigma' \subseteq \Sigma$,
$\proj{\Sigma'}(L)=\{ \proj{\Sigma'}(w) \ | \ w \in L\}$.

\subsection{Timed Automata}
Timed automata (TA) are finite automata extended with real-valued clocks to
specify timing constraints between occurrences of events.  For a
detailed presentation of the fundamental results for timed automata,
the reader is referred to the seminal paper of R.~Alur and
D.~Dill~\cite{AlurDill94}.
\noindent\begin{definition}[Timed Automaton]\label{def-ta} 
  A \emph{Timed Automaton} $A$ is a tuple $(L,$ $l_0,$
  $X,\Sigma_\tauac, E, \inv, F, R)$ where:
$L$ is a finite set of  \emph{locations}; 
$l_0$ is the \emph{initial location};
$X$ is a finite set of \emph{clocks};
$\Sigma$ is a finite set of \emph{actions}; 
$E \subseteq L \times\calC(X) \times \Sigma_\tauac \times 2^X \times
L$ is a finite set of \emph{transitions}; for
$(\ell,g,a,r,\ell') \in E$, $g$ is the \emph{guard}, $a$ the \emph{action},
and $r$ the \emph{reset} set;
$\inv \in \calC(X)^L$ associates with each location an
  \emph{invariant}; as usual we require the invariants to be
  conjunctions of constraints of the form $x \preceq c$ with $\preceq \in
  \{<,\leq\}$.
  $F \subseteq L$ and $R \subseteq L$ are respectively the
  \emph{final} and \emph{repeated} sets of locations. \endef
\end{definition}
A \emph{state} of $A$ is a pair $(\ell,v) \in L \times \setR_{\geq
  0}^X$.
%
%
A \emph{run} $\varrho$ of $A$ from $(\ell_0,v_0)$ is a (finite or
infinite) sequence of alternating \emph{delay} and \emph{discrete}
moves:
\begin{eqnarray*}
  \varrho & = & (\ell_0,v_0) \xrightarrow{\delta_0} (\ell_0,v_0 + \delta_0)
  \xrightarrow{a_0} (\ell_1,v_1) \cdots  \\ & & \cdots \xrightarrow{a_{n-1}} (\ell_n,v_n)
  \xrightarrow{\delta_n} (\ell_n,v_n+ \delta_n) \cdots 
\end{eqnarray*}
\st for every $i \geq 0$:
\begin{itemize}
\item $v_i + \delta \models \inv(\ell_i)$ for $0 \leq \delta \leq \delta_i$;
\item there is some transition $(\ell_i,g_i,a_i,r_i,\ell_{i+1}) \in E$
  \st: ($i$) $v_i + \delta_i \models g_i$ and ($ii$)
  $v_{i+1}=(v_i+\delta_i)[r_i]$.
\end{itemize}
The set of finite (resp. infinite) runs 
from a state $s$ is denoted $\runs^*(s,A)$ (resp. $\runs^\omega(s,A)$)
and we define $\runs^*(A)=\runs^*((l_0,\vect{0}),A)$,
$\runs^\omega(A)=\runs^\omega((l_0,\vect{0}),A)$  
and finally $\runs(A)=\runs^*(A) \cup \runs^\omega(A)$.  If $\varrho$
is finite and ends in $s_n$, we let $\last(\varrho)=s_n$.  Because of
the den\-se\-ness of the time domain, the transition graph of $A$ is
infinite (uncountable number of states and delay edges).
The \emph{trace}, $\trace(\varrho)$, of a run $\varrho$ is the timed
word $\proj{\Sigma}(\delta_0 a_0 \delta_1 a_1 \cdots a_n \delta_n
\cdots)$.  We let
$\dur(\varrho)=\dur(\trace(\varrho))$.  For $V \subseteq \runs(A)$, we
let $\Trace(V)=\{\trace(\varrho) \ | \ \textit{ $\varrho \in V$}\}$.

A finite (resp. infinite) timed word $w$ is \emph{accepted} by $A$ if
it is the trace of a run of $A$ that ends in an $F$-location (resp. a
run that reaches infinitely often an $R$-location).  $\lang^*(A)$
(resp. $\lang^\omega(A)$) is the set of traces of finite
(resp. infinite) timed words accepted by $A$, and $\lang(A)=\lang^*(A)
\cup \lang^\omega(A)$ is the set of timed words accepted by $A$.
In the sequel we often omit the sets $R$ and $F$ in TA and this
implicitly means $F=L$ and $R=\emptyset$.

A timed automaton $A$ is \emph{deterministic} if there is no $\tauac$
labelled transition in $A$, and if, whenever $(\ell,g,a,r,\ell')$ and
$(\ell,g',a,r',\ell'')$ are transitions of $A$, $g \wedge g' \equiv
\false$.  $A$ is \emph{complete} if from each state $(\ell,v)$, and
for each action $a$, there is a transition $(\ell,g,a,r,\ell')$ such
that $v \models g$.  We note \dta the class of deterministic timed
automata.

\subsection{Region Graph of a TA}
The \emph{region graph} $\rg(A)$ of a TA $A$ is a finite quotient of
the infinite graph of $A$ which is time-abstract bisimilar to
$A$~\cite{AlurDill94}.  It is a finite automaton (FA) on the alphabet
$E'= E \cup \{\tauac\}$. The states of $\rg(A)$ are pairs $(\ell,r)$
where $\ell \in L$ is a location of $A$ and $r$ is a \emph{region} of
$\setR_{\geq 0}^X$. More generally, the edges of the graph are tuples
$(s,t,s')$ where $s,s'$ are states of $\rg(A)$ and $t \in E'$.
Genuine unobservable moves of $A$ labelled $\tauac$ are labelled by
tuples of the form $(s,(g,\tauac,r),s')$ in $\rg(A)$.
An edge $(g,\lambda,R)$ in the region graph corresponds to a discrete
transition of $A$ with guard $g$, action $\lambda$ and reset set $R$.
A $\tauac$ move in $\rg(A)$ stands for a delay move to the
time-successor region.  The initial state of $\rg(A)$ is
$(l_0,\vect{0})$.  A final (resp. repeated) state of $\rg(A)$ is a
state $(\ell,r)$ with $\ell \in F$ (resp. $\ell \in R$).  A
fundamental property of the region graph~\cite{AlurDill94} is:
\begin{theorem}[\cite{AlurDill94}] \label{thm-alur}
  $\lang(\rg(A))=\untimed(\lang(A))$. 
\end{theorem}
  The (maximum) size of the region graph is exponential in the number
  of clocks and in the maximum constant of the automaton $A$
  (see~\cite{AlurDill94}): $|\rg(A)|=|L|\cdot |X|! \cdot 2^{|X|} \cdot
  K^{|X|}$ where $K$ is the largest constant used in $A$.

\subsection{Product of TA}
\begin{definition}[Product of two TA] \label{def-prod-sync} Let
  $A_i=(L_i,l_0^i,X_i,$ $\Sigma^i_{\tauac},E_i,\inv_i)$ for $i \in\{1,2\}$,
  be two TA \st $X_1 \cap X_2 = \emptyset$.  The \emph{product} of
  $A_1$ and $A_2$ is the TA $A_1 \times
  A_2=(L,l_0,X,\Sigma_{\tauac},$ $E,\inv)$ given by:
$L=L_1 \times L_2$;
$l_0=(l_0^1,l_0^2)$;
$\Sigma=\Sigma^1 \cup \Sigma^2$;
$X = X_1 \cup X_2$; and
$E \subseteq L \times \calC(X) \times \Sigma_\tauac \times 2^X \times
    L$ and
    $((\ell_1,\ell_2),g_{1,2},\sigma,r,(\ell'_1,\ell'_2)) \in E$
    if:
    \begin{itemize}
    \item either $\sigma \in (\Sigma_1 \cap \Sigma_2) \setminus
      \{\tauac \}$, and ($i$) $(\ell_k,g_k,\sigma,r_k,\ell'_k) \in
      E_k$ for $k=1$ and $k=2$; ($ii$) $g_{1,2} = g_1 \wedge g_2$ and
      ($iii$) $r=r_1 \cup r_2$;
    \item or for $k=1$ or $k=2$, $\sigma \in (\Sigma_k \setminus
      \Sigma_{3-k}) \cup \{\tauac\}$, and ($i$)
      $(\ell_k,g_k,\sigma,r_k,\ell'_k) \in E_k$; ($ii$) $g_{1,2}=g_k$
      and ($iii$) $r=r_k$;
    \end{itemize}
    and finally $\inv(\ell_1,\ell_2)= \inv(\ell_1) \wedge
    \inv(\ell_2)$.
     \endef
\end{definition}


\section{Fault Diagnosis Problems \& Known Results}\label{sec-fd}

\subsection{The Model}
To model timed systems with faults, we use timed automata on the
alphabet $\Sigma_{\tauac,f}=\Sigma_{\tauac}\cup \{f\}$ where $f$ is
the \emph{faulty} (and unobservable) event. We only consider one type
of fault, but the results we give are valid for many types of
faults $\{f_1,f_2, \cdots,f_n\}$: indeed solving the many types
diagnosability problem amounts to solving $n$ one type diagnosability
problems~\cite{yoo-lafortune-tac-02}.
The observable events are given by $\Sigma_o \subseteq \Sigma$ and
$\tauac$ is always unobservable.

The system we want to supervise is given by a TA
$A=(L,l_0,$$X,\Sigma_{\tauac,f},E,\inv)$. Fig.~\ref{fig-ex-diag1}
gives an example of such a system.  Invariants in the automaton $\calA$
 are written within square brackets as in $[x \leq 3]$.
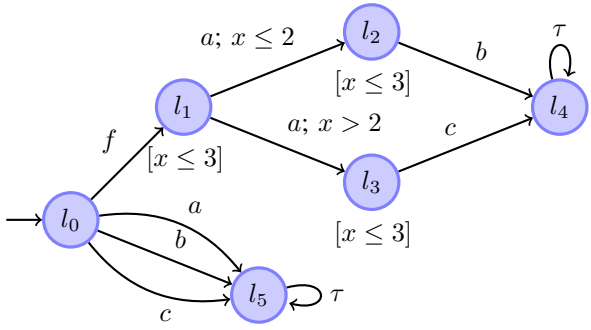
\begin{figure}[hbtp]
  \centering
  \begin{tikzpicture}[thick,node distance=1cm and 2.5cm]%
    \node[state,initial] (l0) {$l_0$}; 
    \node[state] (l1) [above right=of l0,xshift=-1cm,yshift=0.5cm,label=-87:{$[x \leq 3]$}] {$l_1$}; 
    \node[state] (l2) [above right=of l1,label=-87:{$[x \leq 3]$}] {$l_2$};
    \node[state] (l3) [below right=of l1,label=-87:{$[x \leq 3]$}] {$l_3$};
    \node[state] (l4) [above right=of l3] {$l_4$};
    \node[state] (l5) [below right=of l0] {$l_5$};
    \path[->] (l0) edge node[pos=0.5] {$f$} (l1) 
                edge[bend angle=30,bend left] node[pos=0.5]  {$a$} (l5)
                edge node[pos=0.5]  {$b$} (l5)
                edge[bend angle=30,bend right] node[swap,pos=0.7]  {$c$} (l5)
              (l1) edge [pos=0.7] node  {$a$; $x \leq 2$} (l2)
              (l1) edge  node  {$a$; $x > 2$} (l3)
              (l2) edge  [] node  {$b$} (l4)
              (l3) edge  [] node  {$c$} (l4)
              ;
    \path[->] (l4) edge[loop above] node {$\tauac$} (l4);
    \path[->] (l5) edge[loop right] node {$\tauac$} (l5);
  \end{tikzpicture}
\caption{The Timed Automaton $\calA$}
\label{fig-ex-diag1}
\end{figure}

\noindent Let $\Delta \in \setN$. A run of $A$
\begin{eqnarray*}
  \varrho & = & (\ell_0,v_0) \xrightarrow{\delta_0} (\ell_0,v_0 + \delta_0)
  \xrightarrow{a_0} (\ell_1,v_1) \cdots \\ 
  & & \cdots \xrightarrow{a_{n-1}} (\ell_n,v_n)
  \xrightarrow{\delta_n} (\ell_n,v_n+ \delta) \cdots
\end{eqnarray*}
is $\Delta$-faulty if: (1) there is an index $i$ \st $a_i=f$ and (2)
the duration of the run $\varrho'=(\ell_{i},v_i)
\xrightarrow{\delta_{i}} \cdots \xrightarrow{\delta_n}
(\ell_n,v_n+\delta_n) \cdots$ is larger than $\Delta$.  We let
$\faulty_{\geq \Delta}(A)$ be the set of $\Delta$-faulty runs of $A$.
Note that by definition, if $\Delta' \geq \Delta$ then $\faulty_{\geq
  \Delta'}(A) \subseteq \faulty_{\geq \Delta}(A)$. We let
$\faulty(A)=\cup_{\Delta \geq 0}\faulty_{\geq \Delta}(A)=\faulty_{\geq
  0}(A)$ be the set of faulty runs of $A$, and $\nonfaulty(A)=\runs(A)
\setminus \faulty(A)$ be the set of non-faulty runs of $A$.
Moreover we use  
$$\faulty^{\textit{tr}}_{\geq
  \Delta}(A)=\Trace(\faulty_{\geq \Delta}(A))$$ and
$$\nonfaulty^{\textit{tr}}(A)=\Trace(\nonfaulty(A))$$ 
which are the traces\footnote{Notice that $\trace(\varrho)$ erases
  $\tauac$ and $f$.} of $\Delta$-faulty and non-faulty runs of $A$.

\subsection{Diagnosers}
The purpose of fault diagnosis is to detect a fault as soon as
possible.  Faults are unobservable and only the events in $\Sigma_o$
can be observed as well as the time elapsed between these events.
Whenever the system generates a timed word $w$, the observer can only
see $\proj{\Sigma_o}(w)$.  If an observer can detect faults in this
way it is called a \emph{diagnoser}.  A diagnoser must detect a fault
within a given delay $\Delta \in \setN$.

\begin{definition}[$(\Sigma_o,\Delta)$-Diagnoser]\label{def-diag}
  Let $A$ be a TA over the alphabet $\Sigma_{\tauac,f}$, $\Sigma_o
  \subseteq \Sigma$ and $\Delta \in \setN$.  A
  \emph{$(\Sigma_o,\Delta)$-diagnoser} for $A$ is a mapping $D:
  \tw^*(\Sigma_o)\rightarrow \{0,1\}$ such that:
  \begin{itemize}
  \item for each $\varrho \in \nonfaulty(A)$,
    $D(\proj{\Sigma_o}(\varrho))=0$,
  \item for each $\varrho \in \faulty_{\geq \Delta}(A)$,
    $D(\proj{\Sigma_o}(\varrho))=1$. \endef
  \end{itemize}
\end{definition}
$A$ is $(\Sigma_o,\Delta)$-diagnosable if there exists a
$(\Sigma_o,\Delta)$-diagnoser for $A$. $A$ is $\Sigma_o$-diagnosable
if there is some $\Delta \in \setN$ \st $A$ is
$(\Sigma_o,\Delta)$-diagnosable.

\begin{example}
  The TA $\calA$ in Fig.~\ref{fig-ex-diag1} with
  $\Sigma=\Sigma_o=\{a,b,c\}$ is $(\Sigma,3)$-diagnosable.  For the
  timed words with an $a$ followed by either a $b$ or a $c$ a fault
  must have occurred.  Otherwise no fault should be reported.
%
  If $\Sigma_o=\{b\}$, in $\calA$ there are two runs:
  \begin{eqnarray*}
    \rho_1 \!\!\!  & = & \!\!\! (l_0,0) \xrightarrow{\; f \; } (l_1,0) 
    \xrightarrow{\; a \; } (l_2,0) 
    \xrightarrow{\; 3 \; } (l_2,3) 
    \xrightarrow{\; b\; } (l_4,3) \cdots  \\
    \rho_2 \!\!\! & = & \!\!\! (l_0,0) 
    \xrightarrow{3} (l_0,3) \xrightarrow{\ b \ } (l_5,3) \cdots  
  \end{eqnarray*}
  that satisfy $\trace(\rho_1)=\trace(\rho_2)$, and thus $\calA$ is
  not $(\{b\},3)$-dia\-gnosable. To diagnose a fault in $\calA$, $a$
  must be observed. \endex
\end{example}

\subsection{Classical Diagnosis Problems}
\noindent
Assume $A=(L,\ell_0,X,\Sigma_{\tauac,f},E,\inv)$ is a TA
.  The
classical fault diagnosis problems are the following: 
\begin{prob}[Bounded or $\Delta$-Diagnosability] \label{prob-delta-diag} \mbox{} \\
  \textsc{Inputs:} A  TA $A$, $\Sigma_o \subseteq \Sigma$, and $\Delta \in \setN$. \\
  \textsc{Problem:} Is $A$ $(\Sigma_o,\Delta)$-diagnosable?
\end{prob}
\begin{prob}[Diagnosability] \label{prob-diag} \mbox{} \\
  \textsc{Inputs:} A TA $A$ and $\Sigma_o \subseteq \Sigma$. \\
  \textsc{Problem:} Is $A$ $\Sigma_o$-diagnosable?
\end{prob}
\begin{prob}[Maximum delay] \label{prob-delay} \mbox{} \\
  \textsc{Inputs:} A TA $A$ and and $\Sigma_o \subseteq \Sigma$. \\
  \textsc{Problem:} If $A$ is $\Sigma_o$-diagnosable, what is the
  minimum $\Delta$ \st $A$ is $(\Sigma_o,\Delta)$-diagnosable ?
\end{prob}
\smallskip

According to Definition~\ref{def-diag}, $A$
is $\Sigma_o$-diagnosable, \ssi, there is some $\Delta \in \setN$ \st
$A$ is $(\Sigma_o,\Delta)$-diagnosable. Thus $A$ is not $\Sigma_o$-diagnosable
\ssi $\forall \Delta \in \setN$, $A$ is not  $(\Sigma_o,\Delta)$-diagnosable.
%
Moreover a trace based definition of
$(\Sigma_o,\Delta)$-diagnosability can be stated as\footnote{This
  definition does not take into account \emph{Zeno} runs; this is not
  difficult to add and the reader is referred to~\cite{cassez-cdc-09}
  for more details.}: $A$ is $(\Sigma_o,\Delta)$-diagnosable \ssi
\begin{equation}
  \proj{\Sigma_o}(\faulty_{\geq \Delta}^{\trace}(A)) \cap
  \proj{\Sigma_o}(\nonfaulty^{\trace}(A)) = \emptyset \mathpunct. \label{eq-base}
\end{equation}
This gives a necessary and sufficient condition for non
$\Sigma_o$-diagnosability:
\begin{eqnarray}
\label{eq-diagnos2}
\hskip0em\text{$A$ is not $\Sigma_o$-diagnosable} &  \hskip-1.2em \iff  
& \hskip-1.3em 
 \begin{cases}
   \forall \Delta \in \setN, \\
   \quad \exists \rho \in \nonfaulty(A) \\
   \quad \exists \rho' \in \faulty_{\geq \Delta}(A) \textit{ \st }  \\
   \quad\;\;\proj{\Sigma_o}(\rho) = \proj{\Sigma_o}(\rho') \mathpunct,
 \end{cases}
\end{eqnarray}
or in other words, there is no pair of runs $(\rho_1,\rho_2)$ with
$\rho_1 \in \faulty_{\geq \Delta}(A)$, $\rho_2 \in \nonfaulty(A)$ the
$\Sigma_o$-traces of which are equal.

Complexity results for the diagnosis problems on timed automata were
established in~\cite{tripakis-02} (see~\cite{cassez-cdc-09} for a
comprehensive study) and
Problems~\ref{prob-delta-diag}--\ref{prob-delay} are PSPACE-complete
(note that PSPACE-completeness already holds for $\Sigma_o=\Sigma$).

\section{Sensor Minimization with Static Observers}
\label{sec-static}
In this section, we extend the results of~\cite{cassez-acsd-07} to
systems given by TA.

\begin{prob}[Minimum Cardinality Set] \label{prob-static-minimum} \mbox{} \\
  \textsc{Inputs:} A  TA $A=(L,\ell_0,X,\Sigma_{\tauac,f},E,\inv)$ and $n \in \setN$. \\
  \textsc{Problem:}
  \begin{itemize}
  \item[(A)] Is there any set $\Sigma_o \subseteq \Sigma$, with
    $|\Sigma_o| =n$ \st $A$ is $\Sigma_o$-diagnosable ?
  \item[(B)] If the answer to~(A) is ``yes'', compute the minimum
    value for $n$.
  \end{itemize}
\end{prob}
\begin{theorem}
  Problem~\ref{prob-static-minimum} is PSPACE-complete.
\end{theorem}
\begin{proof}
  PSPACE-easiness for (A) can be established as follows: guess a set
  $\Sigma_o$ with $|\Sigma_o| =n$ and check (in PS\-PACE) whether $A$ is
  $\Sigma_o$-diagnosable. This proves  NPSPACE and thus in
  PSPACE.  PSPACE-hardness follows from the reduction of
  Problem~\ref{prob-diag} to Problem~\ref{prob-static-minimum}.(A)
  with $n=|\Sigma|$.  This establishes PSPACE-completeness for (A).
  Computing the minimum $n$ can be done using a binary search
  (dichotomy) and thus (B) is also in PSPACE.
\end{proof}

\medskip The previous results also hold in a more general setting
using \emph{masks}.  Masks are useful to capture the notion of
\emph{distinguishability} among observable events.  Indeed, there are
cases where two events $a$ and $b$ are observable but not
distinguishable, that is, the diagnoser knows that $a$ or $b$
occurred, but not which of the two. This is not the same as
considering $a$ and $b$ to be unobservable, since in that case the
diagnoser would not be able to detect the occurrence of $a$ or $b$.
Distinguishability of events is captured by the notion of a
\emph{mask}~{\cite{Varaiyaetal88}}.

\begin{definition}[Mask]\label{def-mask}
  A \emph{mask} $(M,n)$ (of size $n$) over $\Sigma$ is a total,
  surjective function $M: \Sigma \rightarrow
  \{\mathbf{1},\cdots,\mathbf{n}\} \cup \{\varepsilon\}$. \endef
\end{definition}
$M$ induces a morphism $M^* : \tw^*(\Sigma) \rightarrow
\tw^*(\{\mathbf{1},\cdots,\mathbf{n}\})$, where
$M^*(\varepsilon)=\varepsilon$ and $M^*(a.\rho)=M(a).M^*(\rho)$, for
$a\in\Sigma$ and $\rho\in\Sigma^*$.  For example, if
$\Sigma=\{a,b,c,d\}$, $n=2$ and $M(a)=M(d)=\mathbf{1}$,
$M(c)=\mathbf{2}$, $M(b)=\motvide$, then we have $M^*(a\ 0.4 \ b \ 0.2
\ c \ 1.1 \ b \ 0.7 \ d) = \mathbf{1} \ 0.6 \ \mathbf{2} \ 1.8\
\mathbf{1}$.

\begin{definition}[$(M,n),\Delta)$-diagnoser]\label{def-mask-diag} Let
  $(M,n)$ be a mask over $\Sigma$. A mapping $D:
  \tw^*(\{\mathbf{1},\cdots,\mathbf{n}\}) \rightarrow \{0,1\}$ is a
  \emph{$((M,n),\Delta)$-diagnoser} for $A$ if:
  \begin{itemize}
  \item for each $\rho \in \nonfaulty(A)$,
    $D(M^*(\trace(\rho)))=0$;
  \item for each $\rho \in \faulty_{\geq k}(A)$,
    $D(M^*(\trace(\rho)))=1$. \endef
  \end{itemize}
\end{definition}

$A$ is $((M,n),\Delta)$-diagnosable if there is a
$((M,n),\Delta)$-diagnoser for $A$. $A$ is said to be
$(M,n)$-diagnosable if there is some $\Delta$ such that $A$ is
$((M,n),\Delta)$-diagnosable.  Given a mask $(M,n)$ and $A$, checking
whether $A$ is $(M,n)$-diagnosable can be done in PSPACE: it suffices
to replace each event $a \in \Sigma$ by $M(a)$ and check for
diagnosability. It is PSPACE-complete as using an identity mask of
cardinality $|\Sigma|$ solves Problem~\ref{prob-diag}.

The counterpart of Problem~\ref{prob-static-minimum} with masks is the
following:

\begin{prob}[Minimum Cardinality Mask] \label{prob-static-mask} \mbox{} \\
  \textsc{Inputs:} A  TA $A=(L,\ell_0,X,\Sigma_{\tauac,f},E,\inv)$ and $n \in \setN$. \\
  \textsc{Problem:}
  \begin{itemize}
  \item[(A)] Is there any mask $(M,n)$, \st $A$ is $(M,n)$-diagnosable?
  \item[(B)] If the answer to~(A) is ``yes'', compute the minimum value for $n$.
  \end{itemize}
\end{prob}
\begin{theorem}\label{thm-mask}
  Problem~\ref{prob-static-mask} is PSPACE-complete.
\end{theorem}
\begin{proof}
  PSPACE-easiness is proved by: 1) guessing a mask $(M,n)$ and
  checking (in PSPACE) that $A$ is $(M,n)$-diagnosable.
  PSPACE-hardness is proved as follows.  If there is a mask $(M,n)$
  with $n=|\Sigma|$ \st $A$ is $(M,n)$-diagnosable, then, as $M$ is
  surjective, it must be the case that $M$ is a one-to-one mapping
  from $\Sigma$ to $\{\mathbf{1},\cdots,\mathbf{n}\}$. It follows that
  $A$ is $\Sigma$-diagnosable.  Conversely, assume
  $\Sigma=\{a_1,\cdots,a_n\}$.  If $A$ is $\Sigma$-diagnosable then
  there is a mask $(M,|\Sigma|)$ with $M(a_i)=i$ \st $A$ is
  $(M,|\Sigma|)$-diagnosable.  Hence
  Problem~\ref{prob-static-mask}.(A) is PSPACE-complete.
  Problem~\ref{prob-static-mask}.(B) can be solved in PSPACE as well
  using a binary search. It is not difficult to reduce reachability
  for TA with one action to checking whether there is a mask of size
  $1$ and thus Problem~\ref{prob-static-mask}.(B) is PSPACE-complete.
\end{proof}
\begin{remark}
  The assumption that a mask is surjective can be lifted still
  preserving Theorem~\ref{thm-mask}. Indeed, if there is a mask
  $(M,|\Sigma|)$ \st $A$ is $(M,|\Sigma|)$-diagnosable and $M$ is not
  surjective, then we can build $(M',|\Sigma|)$ with $M'$ surjective
  \st $A$ is $(M',|\Sigma|)$-diagnosable (intuitively, $M'$ is more
  discriminating than $M$ and has a greater distinguishing power).
\end{remark}

\section{Sensor Minimization with Dynamic Observers}
\label{sec-dynamic}
The use of \emph{dynamic observers} was already advocated for DES
in~\cite{cassez-acsd-07,cassez-fi-08}.  We start with an example that
shows that dynamically choosing what to observe can be even more
efficient using timing information.

\begin{example}
  Let $\calA$ be the automaton of Figure~\ref{fig-ex-diag1}.  To
  diagnose $\calA$, we can use a \emph{dynamic observer} that
  switches $a$, $b$ and $c$-sensors on/off.  If we do not measure time,
  to be able to detect faults in $\calA$, we have to switch the $a$
  sensor on at the beginning. When an $a$ has occurred, we must be
  ready for either an $b$ or a $c$ and therefore, switch on the $b$
  and $c$ sensors on. A dynamic observer must thus first observe
  $\{a\}$ and after an occurrence of $a$, observe $\{b,c\}$.

  If the observer can measure time using a clock, say $y$, it can
  first switch the $a$ sensor on. If an $a$ occurs when $y \leq 2$,
  then switch the $b$ sensor on and if $y >2$ switch the $c$ sensor
  on. This way the observer never has to observe more than event at
  each point in time. \endex
\end{example}

\subsection{Dynamic Observers}
The choice of the events to observe can depend on the choices the
observer has made before and on the observations (event, time-stamp) it
has made. Moreover an observer may have \emph{unbounded} memory.  The
following definition extends the notion of observers introduced
in~\cite{cassez-acsd-07} to the timed setting.
 
\begin{definition}[Observer]\label{def-observer2}
  An \emph{observer} \obs over $\Sigma$ is a \emph{deterministic and
    complete} timed automaton $\obs=(N,n_0,Y,$
  $\Sigma,\delta,\inv_{\true})$ together with a mapping $O: N
  \rightarrow 2^\Sigma$, where $N$ is a (possibly infinite) set of
  locations, $n_0\in N$ is the initial location, $\Sigma$ is the set
  of observable events, $\delta : N \times \Sigma \times \calC(Y)
  \rightarrow N \times 2^Y$ is the transition function (a total
  function), and $O$ is a labeling function that specifies the set of
  events that the observer wishes to observe when it is at location
  $n$. The invariant\footnote{In the sequel, we omit the invariant
    when a TA is an observer, and replace it by the mapping $O$.}
  $\inv_{\true}$ maps every location to $\true$, implying that an
  observer cannot prevent time from elapsing. We require that, for any
  location $n$ and any $a\in\Sigma$, if $a\not\in O(n)$ then
  $\delta(n,a,\cdot)=(n,\varnothing)$: this means the observer does
  not change its location nor resets its clocks when an event it has
  chosen not to observe occurs.  \endef
\end{definition}
As an observer is deterministic we let $\delta(n_0,w)$ 
denote the state $(n,v)$ reached after reading the timed word $w$ and
$O(\delta(n_0,w))$ is the set of events $\obs$ observes after $w$.

\noindent An observer defines a {\em transducer} which is a mapping
$\sem{\obs} : \tw^*(\Sigma)\rightarrow \tw^*(\Sigma)$. Given a word
$w$, $\sem{\obs}(w)$ is the out\-put of the transducer on $w$.  It is
called the \emph{observation} of $w$ by the observer \obs.

\subsection{Diagnosability with Dynamic Observers}
\begin{definition}[$(\obs,\Delta)$-diagnoser] \label{def-obsk-diag}
  Let $A$ be a TA over $\Sigma_{\tauac,f}$ and \obs be an observer
  over $\Sigma$. $D:\tw^*(\Sigma) \rightarrow \{0,1\}$ is an
  \emph{$(\obs,\Delta)$-diagnoser} for $A$ if:
  \begin{itemize}
  \item $\forall \rho \in \nonfaulty(A)$,
    $D(\sem{\obs}(\trace(\rho)))=0$ and
  \item $\forall \rho \in \faulty_{\geq \Delta}(A)$,
    $D(\sem{\obs}(\trace(\rho)))=1$. \endef
  \end{itemize}
\end{definition}
$A$ is $(\obs,\Delta)$-diagnosable if there is an
$(\obs,\Delta)$-diagnoser for $A$. $A$ is \obs-diagnosable if there is
some $\Delta$ such that $A$ is $(\obs,\Delta)$-diagnosable.

We now show how to check $\obs$-diagnosability when the observer
$\obs$ is a DTA.
\begin{prob}[Deterministic Timed Automata Observers] \label{prob-dynamic-ta} \mbox{} \\
  \textsc{Inputs:} A TA $A=(L,\ell_0,X,\Sigma_{\tauac,f},E,\inv)$ and
  an observer given by a \dta
  $\obs=(N,n_0,Y,\Sigma,\delta,O)$. \\
  \textsc{Problem:}
  \begin{itemize}
  \item[(A)] Is $A$ $\obs$-diagnosable?
  \item[(B)] If the answer to~(A) is ``yes'', compute the minimum $\Delta \in \setN$
    \st $A$ is $(\obs,\Delta)$-diagnosable.
  \end{itemize}
\end{prob}

\begin{theorem}
  Problem~\ref{prob-dynamic-ta} is PSPACE-complete.
\end{theorem}

\begin{proof}
  PSPACE-hardness follows from the fact that taking an observer which
  always observes $\Sigma_o \subseteq \Sigma$ solves
  Problem~\ref{prob-diag}.  We prove that
  Problem~\ref{prob-dynamic-ta} is in PSPACE.  The following
  construction is an extension of the one for DES~\cite{cassez-fi-08}.
  Recall that $\obs$ is complete.  Define the timed automaton $A
  \otimes \obs=(L \times N,(\ell_0,n_0),X \cup
  Y,\Sigma_{\tauac,f},\rightarrow,\inv_\otimes)$ as follows:
  $\inv_\otimes(\ell,n)=\inv(\ell)$ and the transition relation
  $\rightarrow$ is given by:
\begin{itemize}
\item $(\ell,n) \xrightarrow{\,(g \wedge g',\beta,R \cup Y')\,}
  (\ell',n')$ iff $\exists \lambda \in \Sigma$ \st $\ell
  \xrightarrow{\,(g,\lambda,R)\,} \ell'$,
  $(n',Y')=\delta(n,\lambda,g')$ and $\beta=\lambda$ if $\lambda \in
  O(n)$, $\beta=\tauac$ otherwise;
\item $(\ell,n) \xrightarrow{\,(g,\lambda,R)\,} (\ell',n)$ iff 
  $\exists \lambda \in \{\tauac,f\}$ \st $\ell
  \xrightarrow{\,(g,\lambda,R)\,} \ell'$.
\end{itemize}
The TA $A \otimes \obs$ is an unfolding of $A$ which reveals what is
observable at each product location.

From the previous construction, it follows that: for each $\Delta \in
\setN$, $A$ is $(\obs,\Delta)$-diagnosable iff $A \otimes \obs$ is
$(\Sigma,\Delta)$-diagnosable.  As the size of $A \otimes \obs$ is
$|A| \times |\obs|$, we can solve Problem~\ref{prob-dynamic-ta}.(A) in
PSPACE.  Problem~\ref{prob-dynamic-ta}.(B) can also be solved using a
binary search, in PSPACE.
\end{proof}

\subsection{Synthesis of the Most Permissive Dynamic Diagnoser}
In this section we address the problem of \emph{synthesizing} a \dta
dynamic observer which ensures
diagnosability. Following~\cite{cassez-fi-08}, we want to compute a
\emph{most permissive} observer ($\varnothing$ if none exists), which
gives a representation of all the good observers.  Indeed, checking
whether there exists a DTA observer $\obs$ \st $A$ is
$\obs$-diagnosable is not an interesting problem: it suffices to check
that $A$ is $\Sigma$-diagnosable as the DTA observer which observes
$\Sigma$ continuously will be a solution.

When synthesizing (deterministic) timed automata, an important issue
is the amount of \emph{resources} the timed automaton can use: this
can be formally defined~\cite{BDMP-cav-2003} by the (number of) clocks, $Z$,
that the automaton can use, the maximal constant $\max$, and a
\emph{granularity} $\frac{1}{m}$. As an example, a TA of resource
$\mu=(\{c,d\},2,\frac{1}{3})$ can use two clocks, $c$ and $d$, and the
clocks constraints using the rationals $-2 \leq k/m \leq 2$ where $k
\in \setZ$ and $m=3$.  A \emph{resource} $\mu$ is thus a triple
$\mu=(Z,\max,\frac{1}{m})$ where $Z$ is finite set of clocks, $\max
\in \setN$ and $\frac{1}{m} \in \setQ_{>0}$ is the \emph{granularity}.
\dtamu is the class of \dta of resource $\mu$.
\begin{remark}
  Notice that the number of locations of the \dta in \dtamu is not
  bounded and hence this family has an infinite (yet countable) number
  of elements.
\end{remark}

We now focus on the following problem :
\begin{prob}[Most Permissive Dynamic $\Delta$-Diagnoser] \label{prob-dynamic-synth} \mbox{} \\
  \textsc{Inputs:} A TA $A=(L,\ell_0,X,\Sigma_{\tauac,f},E,\inv)$,
  $\Delta \in \setN$,
  and a resource $\mu=(Z,\max,\frac{1}{m})$.\\
  \textsc{Problem:} Compute the set $O$ of all observers in \dtamu,
  \st $A$ is $(\obs,\Delta)$-diagnosable iff $\obs \in O$.
\end{prob}
For DES, the previous problem can be solved by computing a most
permissive observer, and we refer to~\cite{cassez-fi-08} section~5.5
for the formal definition of the most permissive observer. This can be
done in 2EXPTIME~\cite{cassez-fi-08}, and the solution is a reduction
to a safety control problem under partial observation.  For the timed
case, we cannot use the same solution as controller synthesis under
partial observation is undecidable~\cite{BDMP-cav-2003}.  The solution
we present for Problem~\ref{prob-dynamic-synth} is a modification of
an algorithm originally introduced in~\cite{Bouyerfossacs05}. 

\subsection{Fault Diagnosis with DTA~\cite{Bouyerfossacs05}}\label{sec-algo}
In case a TA $A$ is $\Sigma_o$-diagnosable, the diagnoser is a
mapping~\cite{tripakis-02} which performs a state estimate of $A$
after a timed word $w$ is read by $A$.  For DES, it is obtained by
\emph{determinizing} the system, but we cannot always determinize a TA
$A$ (see~\cite{AlurDill94}).  And unfortunately testing whether a
timed automaton is determinizable is
undecidable~\cite{Finkel05,TripakisFolk}.

P.~Bouyer and F.~Chevalier in~\cite{Bouyerfossacs05} considers the
problem of deciding whether there exists a diagnoser which is a DTA
using resources in $\mu$:

\begin{prob}[\dtamu $\Delta$-Diagnoser~\cite{Bouyerfossacs05}] \label{prob-dtamu} \mbox{} \\
  \textsc{Inputs:} A TA $A=(L,\ell_0,X,\Sigma_{\tauac,f},E,\inv)$,
  $\Delta \in \setN$,
  and a resource $\mu=(Z,\max,\frac{1}{m})$.\\
  \textsc{Problem:} Is there any $D \in \text{DTA}_\mu$ \st $A$ is
  $(D,\Delta)$-dia\-gnosable ?
\end{prob}
\begin{theorem}[\cite{Bouyerfossacs05}]
Problem~\ref{prob-dtamu}  is 2EXPTIME-complete.
\end{theorem}

The solution to the previous problem is based on the construction of a
\emph{two-player game}, the solution of which gives the \emph{set} of
all $\text{DTA}_\mu$ diagnosers (the most permissive diagnosers) which
can diagnose $A$ (or $\varnothing$ is there is none).

We recall here the construction of the two-player game.

Let $A=(L,\ell_0,X,\Sigma_{\tauac,f},\rightarrow,\inv)$ be a TA, $\Sigma_o
\subseteq \Sigma$.  Define $A(\Delta)=(L_1 \cup L_2 \cup
L_3,\ell^1_0,X \cup
\{z\},\Sigma_{\tauac,f},\rightarrow_\Delta,\inv_\Delta)$ as follows:
\begin{itemize}
\item $L_i=\{\ell^i, \ell \in L\}$, for $i\in \{1,2,3\}$, \ie $L_i$
  elements are copies of the locations in $L$,
\item $z$ is (new) clock not in $X$,
\item for $\ell \in L$, $\inv(\ell^1)=\inv(\ell)$,
  $\inv(\ell^2)=\inv(\ell) \wedge z \leq \Delta$, and 
  $\inv(\ell^3)=\true$,
\item the transition relation is given by:
  \begin{itemize}
  \item for $i \in \{1,2,3\}$, $\ell^i \xrightarrow{\ (g,a,R)\
    }_\Delta \ell'^i$ if $a \neq f$ and $\ell \xrightarrow{\ (g,a,R)\
    } \ell'$,
  \item for $i \in \{2,3\}$, $\ell^i \xrightarrow{\ (g,f,R)\ }_\Delta
    \ell'^i$ if $a \neq f$ and $\ell \xrightarrow{\ (g,f,R)\ } \ell'$,
  \item $\ell^1 \xrightarrow{\ (g,f,R \cup \{z\})\ }_\Delta
    \ell'^2$ if $a \neq f$ and $\ell \xrightarrow{\ (g,f,R)\ } \ell'$,
  \item $\ell^2 \xrightarrow{\ (z=\Delta,\tauac,\varnothing)\ }_\Delta
    \ell^3$.
  \end{itemize}
\end{itemize}
The previous construction creates $3$ copies of $A$: the system starts
in copy $1$, when a fault occurs it switches to copy $2$, resetting
the clock $z$, and when in copy $2$ (a fault has occurred) it can
switch to copy $3$ after $\Delta$ time units.  We can then define
$L_1$ as the non-faulty locations, and $L_3$ as the $\Delta$-faulty
locations.


Given a resource $\mu=(Y,\max,\frac{1}{m})$ ($X \cap Y =\emptyset$), a
\emph{minimal guard} for $\mu$ is a guard which defines a region of
granularity $\mu$. We define the (symbolic) \emph{universal automaton}
$\calU=(\{0\},\{0\},Y,\Sigma,E_\mu,\inv_\mu)$ by:
\begin{itemize}
\item $\inv_\mu(0)=\true$,
\item $(0,g,a,R,0) \in E_\mu$ for each $(g,a,R)$ \st $a \in \Sigma$,
  $R \subseteq Y$, and $g$ is a minimal guard for $\mu$.
\end{itemize}

$\calU$ is finite because $E_\mu$ is finite.  Nevertheless $\calU$ is
not deterministic because it can choose to reset different sets of
clocks $Y$ for a pair ``(guard, letter)'' $(g,a)$.  To diagnose $A$,
we have to find when a set of clocks has to be reset. This can provide
enough information to distinguish $\Delta$-faulty words from
non-faulty words.

The algorithm of~\cite{Bouyerfossacs05} requires the following steps:
\begin{enumerate}
\item define the region graph $\rg(A(\Delta) \times \calU)$,
\item compute a \emph{projection} of this region graph:
  \begin{itemize}
  \item let $(g,a,R)$ be a label of an edge in $\rg(A(\Delta)
    \times \calU)$,
  \item let $g'$ be the unique minimal guard \st $\sem{g} \subseteq
    \sem{g'}$;
  \item define the projection $p_\calU(g,a,R)$ by $(g',\lambda,R \cap
    Y)$ with $\lambda=a$ if $a \in \Sigma_o$ and
    $p_\calU(g,a,R)=\tauac$ otherwise.
  \end{itemize}
  The projected automaton $p_\calU(\rg(A(\Delta) \times \calU))$ is
  the automaton $\rg(A(\Delta) \times \calU)$ where each label
  $\alpha$ is repla\-ced by $p_\calU(\alpha)$.
\item determinize $p_\calU(\rg(A(\Delta) \times \calU))$ (removing
  $\tauac$ actions) and obtain $H_{A,\Delta,\mu}$,
\item build a two-player safety game $G_{A,\Delta,\mu}$ as follows:
  \begin{itemize}
  \item each transition $s  \xrightarrow{\ (g,a,Y) \ }  s'$
    %
    in $H_{A,\Delta,\mu}$ yi\-elds a transition in $G_{A,\Delta,\mu}$ of
    the form:

    
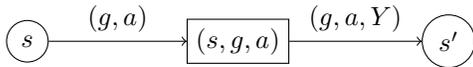
\begin{center}
      \tikz[node distance=1cm and 2.8cm]{ \node[circle,draw] (l0)
        {$s$}; \node[rectangle,draw,right=of l0] (l2) {$(s,g,a)$};
        \node[circle,draw,right=of l2] (l1) {$s'$}; \path[->] (l0)
        edge node {$(g,a)$} (l2) ; \path[->] (l2) edge node
        {$(g,a,Y)$} (l1) ; }
    \end{center}

  \item the round-shaped state are the states of Player~1, whereas the
    square-shaped states are Player~0 states (the choice of the clocks
    to reset).
  \item the $\text{Bad}$ states (for Player~0) are the states of the
    form $\{(\ell_1,r_1),(\ell_2,r_2),\cdots,(\ell_k,r_k)\}$ with both
    a $\Delta$-faulty (in $L_3$) and a non-faulty (in $L_1$) location.
  \end{itemize}
\end{enumerate}
The main results of~\cite{Bouyerfossacs05} are:
\begin{itemize}
\item there is a TA $D \in$ \dtamu \st $A$ is $(D,\Delta)$-diagnosable
  iff Player~0 can win the safety game ``avoid Bad''
  $G_{A,\Delta,\mu}$,
\item it follows that Problem~\ref{prob-dtamu} can be solved in
  2EXPTIME as $G_{A,\Delta,\mu}$ has size doubly exponential in $A$,
  $\Delta$ and $\mu$,
\item the acceptance problem for Alternating Turing machines of exponential
  space can be reduced to Problem~\ref{prob-dtamu} and thus it
  is 2EXPTIME-hard.
\end{itemize}

\subsection{Problem~\ref{prob-dynamic-synth} is in 2EXPTIME} \label{sec-synth}
We now show how to modify the previous algorithm to solve
Problem~\ref{prob-dynamic-synth}, and obtain the following result:
\begin{theorem}
  Problem~\ref{prob-dynamic-synth} can be solved in 2EXPTIME.
\end{theorem}
\begin{proof}
  We modify the previous algorithm as follows:
  \begin{enumerate}
  \item the automaton $\calU$ is defined as follows: each location
    corresponds to a choice of a subset of events to observe. Define
    the (symbolic) \emph{universal automaton}
    $\calU'=(2^\Sigma,2^\Sigma,Y,\Sigma,E_\mu,\inv_\mu)$ by:
\begin{itemize}
\item for $S \in 2^\Sigma$, $\inv_\mu(S)=\true$,
\item $(S,g,a,R,S') \in E_\mu$ for each $S,S' \in 2^\Sigma$, $(g,a,R)$
  \st $a \in \Sigma$, $R \subseteq Y$, and $g$ is a minimal guard for
  $\mu$. 
\end{itemize}
\item when computing $\rg(A(\Delta) \times \calU'))$, the set of
  observable events (step~2 in the algorithm of
  section~\ref{sec-algo}) are defined according to the location $S$ of
  $\calU'$.  Formally, the projection of $a \in \Sigma$ is $a$ if the
  location of $\calU'$ is $S$ and $a \in S$ and $\tauac$ otherwise.
  \end{enumerate}
  The size of $\rg(A(\Delta) \times \calU'))$ is $|L| \cdot
  2^{|\Sigma|} \cdot |X \cup Y|! \cdot K^{|X \cup Y|}$ where $K$ is
  the maximal constant of $A \times \calU'$; it is thus exponential in
  $\mu$ and $\Sigma$.  The determinization is thus doubly exponential
  in $A$, $\mu$ and $\Sigma$. We can then build a new game
  $G'_{A,\Delta,\mu}$ as described in section~\ref{sec-algo} before.
  The proof that the most permissive strategy in the new game
  $G'_{A,\Delta,\mu}$ is the most permissive observer is along the
  lines of the one given in~\cite{Bouyerfossacs05} with minor
  modifications.
  Solving a safety game is linear in the size of the game and thus
  computing the most permissive observer of resource $\mu$ can de done
  in 2EXPTIME.
\end{proof}

\begin{remark}
  In~\cite{Bouyerfossacs05} it is also proved that for Event Recording
  Automata (ERA)~\cite{AFH94} Problem~\ref{prob-dtamu} becomes
  PSPACE-complete.  This result does not carry over in our case, as
  there is still an exponential step with the choice of the sets of
  events to be observed.
\end{remark}

\section{Optimal Dynamic Observers}\label{sec-cost}
In this section we extend the notion of \emph{cost} defined for finite
state observers in~\cite{cassez-fi-08} to the case of timed observers.

\subsection{Weighted/Priced Timed Automata}
Weighted/priced timed automata were introduced in~\cite{BFH+01,ATP01}
and they extend TA with \emph{prices/costs/weights} on the time
elapsing and discrete transitions.

\begin{definition}[Priced Timed Automata]
  A \emph{priced timed auto\-ma\-ton (PTA)} is a pair $(A,\Cost)$
  where $A=(L,\ell_0,X,$ $\Sigma_{\tauac,f},E,\inv)$ is a timed automaton
  and $\Cost$ is a \emph{cost function} which is a mapping from $L
  \cup E$ to $\setN$. \endef
\end{definition}
Let 
\begin{eqnarray*}
  \varrho & = & (\ell_0,v_0) \xrightarrow{\delta_0} (\ell_0,v_0 + \delta_0)
  \xrightarrow{a_0} (\ell_1,v_1) \cdots  \\ & & \cdots \xrightarrow{a_{n-1}} (\ell_n,v_n)
  \xrightarrow{\delta_n} (\ell_n,v_n+ \delta_n) 
\end{eqnarray*}
be a run of $A$. We denote by $e_i=(\ell_i,(g_i,a_i,R_i),\ell_{i+1})$
the discrete transition taken from $(\ell_i,v_i+\delta_i)$ to
$(\ell_{i+1},v_{i+1})$.

The \emph{cost} of the run $\varrho$ is defined by:
\[
\Cost(\varrho) = \Sigma_{i \in 0..n} \Cost(\ell_i) \cdot \delta_i + 
\Sigma_{i \in 0..n-1} \Cost(e_i)\mathpunct.
\]
The \emph{mean cost} of $\varrho$ is defined to be the cost per time
unit and given\footnote{Runs of duration $0$ are not taken into
  account.} by $\MeanCost(\varrho)=\Cost(\varrho)/\dur(\varrho)$.  The
cost of runs of duration $t \in \setR_{>0}$ is defined by
$\MeanCost(t)=\sup \{ \MeanCost(\sem{\obs}(\varrho)) \ | \
\dur(\varrho)=t \}\mathpunct.$ The \emph{maximal mean cost} of
$(A,\cost)$ is $\MeanCost(A)=\limsup_{t\rightarrow \infty}
\MeanCost(t)$.  The minimal mean cost is defined dually and denoted
$\underline{\cost}(A)$.

\subsection{Cost of an Observer}
To select a best or optimal dynamic observer which ensures
$\Delta$-diagnosability, we need to define a metric to compare them.
We extend the one defined in~\cite{cassez-fi-08} for DES to take into
account (real) time elapsing.

Let $A$ be a TA and $\obs$ a DTA observer.  $\obs$ is extended into a
P(D)TA by associating costs with locations and transitions.  The cost
associated with the discrete transitions is the cost of switching on
the sensors for a set of observable events, and the cost of a location
is the cost per time unit of having a set of sensors activated.

Let $\varrho$ be a run of $A$.  As \obs is deterministic (and complete)
there is exactly one run of $\obs$ the trace of which is
$\sem{\obs}(\trace(\varrho))$.  Given $\varrho$, let $\sem{\obs}(\varrho)$ be
this unique run. 
The average cost of the run $\varrho$  observed by \obs is
$\MeanCost(\sem{\obs}(\varrho))$.

Given $t \in \setR_{>0}$, the \emph{maximal mean cost} of runs of
duration $t$ is defined by:
\[
\MeanCost(A,\obs,t)=\sup_{\varrho \in \runs^*(A) \wedge \dur(\varrho)=t} \{ \MeanCost(\sem{\obs}(\varrho))\}\mathpunct.
\]
The \emph{maximal average cost} of the pair $<\!A,\obs\!>$ is defined
\[
\MeanCost(<\!A,\obs\!>)= \limsup_{t\rightarrow \infty} \MeanCost(A,\obs,t)
\mathpunct.
\]
We can then state the following problem:
\begin{prob}[Cost of an Observer] \label{prob-cost-obs} \mbox{} \\
  \textsc{Inputs:} A  TA $A$  and $(\obs,\Cost)$ a PDTA observer. \\
  \textsc{Problem:} Compute $\MeanCost(<\!A,\obs\!>)$.
\end{prob}

\subsection{Computing the Cost of a Given Timed Observer}
The computation of optimal infinite schedules for TA has been
addressed in~\cite{BBL-fmsd06}.  The main result of~\cite{BBL-fmsd06}  is:
\begin{theorem}[Minimal/Maximal Mean
  Cost~\cite{BBL-fmsd06}]\label{thm-pat}
  Given a PTA $A$, com\-puting $\MeanCost$ and $\underline{\cost}$ is
  PSPACE-complete.
\end{theorem}
The definition of the cost of an observer is exactly the defi\-ni\-tion
of the maximal mean cost in~\cite{BBL-fmsd06} and thus:
\begin{theorem}
  Problem~\ref{prob-cost-obs} is PSPACE-complete.
\end{theorem}
\begin{proof}
  PSPACE-easiness follows from Theorem~\ref{thm-pat}: note that
  Theorem~\ref{thm-pat} assumes that the TA is bounded which is not a
  restriction as every TA can be transformed into an equivalent (timed
  bisimilar) bounded TA. For PSPACE-hardness, to compute the maximal
  mean cost of a PDTA $B$, let $A$ be the universal automaton on the
  alphabet of $B$.  Consider $B$ as an observer and solve
  Problem~\ref{prob-cost-obs}.  This solves the maximal mean cost
  computation problem for DTA.  This completes the hardness proof.
\end{proof}

\subsection{Optimal Synthesis Problem}
Checking whether the mean cost of a given observer is less than $k$
requires that we have computed or are given such an observer.  A more
difficult version of Problem~\ref{prob-cost-obs} is to check for the
existence of cheap dynamic observer:
\begin{prob}[Bounded Cost Dynamic Observer] \label{prob-bounded-cost} \mbox{} \\
  \textsc{Inputs:} A TA $A=(L,\ell_0,X,\Sigma_{\tauac,f},E,\inv)$,
  $\Delta \in \setN$,
  $\mu$ a resource and $k \in \setN$. \\
  \textsc{Problem:}
  \begin{itemize}
  \item[(A)] Is there a dynamic observer $D \in$ \dtamu \st $A$ is
    $(D,\Delta)$-diagnosable and $\MeanCost(<\!A,D\!>) \leq k$ ?
  \item[(B)] If the answer to~(A) is ``yes'', compute a witness
    dynamic observer?
  \end{itemize}
\end{prob}
We cannot provide of proof that Problem~\ref{prob-bounded-cost} is
decidable.  However, we give a lower bound for
Problem~\ref{prob-bounded-cost} and later discuss the exact
complexity.
\begin{theorem}
  Problem~\ref{prob-bounded-cost} is 2EXPTIME-hard.
\end{theorem}
\begin{proof}
  We reduce Problem~\ref{prob-dtamu} which is
  2EXPTIME-hard~\cite{Bouyerfossacs05} to Problem~\ref{prob-bounded-cost}.
  Let $A$ be a TA for which we want to check whether there exists a
  DTA observer $D \in$ \dtamu \st $A$ is $(\Delta,D)$-diagnosable.

  Let $\alpha$ be a fresh letter not in $\Sigma$.  Define the
  automaton $B$ depicted on Figure~\ref{fig-reduc-bounded}.  The
  upper part of $B$ generates faulty and non-faulty runs with each
  letter including $\alpha$.  From each location of $A$ (bottom part),
  we add a $\tauac$ transition to the initial state of $B$.  The
  transitions of $A$ are not depicted.
  
  For $B$ to be diagnosable with $\Delta \geq 1$, we must have: 1)
  $\alpha$ always observable and 2) $\Sigma$ always observable.
  Moreover, if $A$ is $(\Delta,\Sigma)$-diagnosable, then $B$ is
  $(\Delta,\Sigma\cup\{\alpha\})$-diagnosable. Conversely, if $B$ is
  $(\Delta,\Sigma\cup\{\alpha\})$-diagnosable, then $B$ is
  $(\Delta,\Sigma)$-diagnosable.  Hence $A$ is
  $(\Delta,\Sigma)$-diagnosable iff $B$ is
  $(\Delta,\Sigma\cup\{\alpha\})$-diagnosable.

  Define the cost of the locations to be $1$, and $0$ for the
  transitions in $B$.  $B$ is diagnosable with a DTA $D \in$ \dtamu
  iff there is a dynamic (yet it has to choose $\Sigma\cup\{\alpha\}$
  continuously) observer $D$ with $\MeanCost(<\!A,D\!>) \leq 1$.
  
  It follows that: there exists a \dtamu diagnoser $D$ \st $A$ is
  $(\Delta,\Sigma)$-diagnosable iff $B$ is $(\Delta,O)$-diagnosable
  with a DTA observer $O \in$ \dtamu and $\MeanCost(<\!A,O\!>) \leq 1$.
\end{proof}
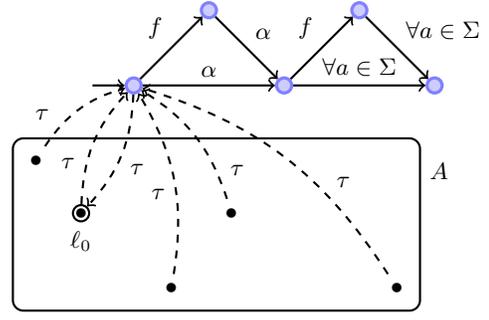
\begin{figure}[t]
  \centering
    \begin{tikzpicture}[thick,node distance=1cm and 1cm]%
    \small
    \tikzset{every state/.style={circle,minimum size=0.2cm,inner sep=0,draw=blue!50,very thick,fill=blue!20},bend angle=20}
    \node[state,initial] (l0) {}; 
    
    \node[state] (l1) [above right=of l0] {}; 
    \node[state] (l2) [right=of l0,xshift=1cm] {};
    \node[state] (l4) [above right=of l2] {};
    \node[state] (l5) [right=of l2,xshift=1cm] {};

    \node (s1) [below left=of l0,xshift=-.3cm] {$\bullet$};
    \node[circle,draw,inner sep=0] (s2) [right=of s1,yshift=-.7cm,xshift=-.4cm,label=-90:{$\ell_0$}] {$\bullet$};
    \node (s3) [below=of s2,xshift=1.2cm] {$\bullet$};
    \node (s4) [right=of s3,xshift=2cm] {$\bullet$};
    \node (s5) [right=of s2,xshift=1cm] {$\bullet$};
   \node[draw=black,inner sep=3pt,rounded corners,thick,fit= (s1) (s2) (s3) (s4) (s5),label=10:{$A$}] (env) {};

    \path[->] (l0) edge node[pos=0.5] {$f$} (l1) 
                edge node[swap,pos=0.5,swap]  {$\alpha$} (l2)
              (l1)      edge [pos=0.5] node  {$\alpha$} (l2)
              (l2) edge  node  {$f$} (l4)
              (l4) edge  node  {$\forall a \in \Sigma$} (l5)
              (l2) edge  [] node  {$\forall a \in \Sigma$} (l5);
    \path[->] (l0) edge[dashed,bend left] node {$\tauac$} (s2);
    \path[->] (s1) edge[dashed,bend left] node[pos=0.2] {$\tauac$} (l0);
    \path[->] (s2) edge[dashed,bend left] node[pos=0.2] {$\tauac$} (l0);
    \path[->] (s3) edge[dashed,bend right] node {$\tauac$} (l0);
    \path[->] (s4) edge[dashed,bend right,swap] node[pos=0.3] {$\tauac$} (l0);
    \path[->] (s5) edge[dashed,bend right,swap] node[pos=0.1] {$\tauac$} (l0);
    
  \end{tikzpicture}
  \caption{Automaton $B$}
  \label{fig-reduc-bounded}
\end{figure}
The status of Problem~\ref{prob-bounded-cost} is clearly unsettled as
the 2EXP\-TIME-hardness result does not imply it is even decidable.  A
solution to this problem would be to mimic the one given for
DES~\cite{cassez-fi-08}: solve a mean payoff \emph{timed} game with a
counterpart of Zwick and Paterson algorithm~\cite{zwick-95} using the
most permissive observers obtained in section~\ref{sec-synth}.  The
type of priced timed games we would have to solve has the following
features: 1) they are turn-based, as one Player picks up (controllable
moves) a set of events to be observed and then hands it over to the
other Player who tries to produce a confusing\footnote{Which is the
  trace of both a faulty and non-faulty run.} run (uncontrollable
moves); 2) they have at least two clocks (one for the system $A$ and
one for the DTA observer); 3) the controllable choices are
\emph{urgent} \ie no time can elapse in Player~1 locations.  We denote
S-PTGA for the class of timed game automata previously defined.

Unfortunately, there is no counterpart of the general result of Zwick
\& Paterson for timed automata.  Only very few results are known for
timed mean payoff
games~\cite{BLMR-fsttcs2006,bflms-formats08,bbjlr-formats08,BFLM-hscc10}
and none of them can be used in our setting.  Nevertheless, due to the
particular nature of the mean payoff price timed game we construct (in
the class S-PTGA), we might be able to compute the optimal choices of
observable events using an algorithm similar to~\cite{BBL-fmsd06}.
Hence we could obtain a 2EXPTIME algorithm for
Problem~\ref{prob-bounded-cost}.

\section{Conclusion}\label{sec-conclu}


The results of the paper are summarized by the line ``TA'' in
Table~\ref{tab-summary}.  The complexity/decidability status of
Problem~\ref{prob-bounded-cost} is left open.  A solution to this
problem would be to solve the following optimization problem on the
class of S-PTGA:
\begin{prob}[Optimal Infinite Schedule in S-PTGA] \mbox{}\\
  \textsc{Inputs:} A S-PTGA $(A,\cost)$, a set of \emph{Bad} states and $k \in \setN$.\\
  \textsc{Problem:} Is there a strategy $f$ for Player~1 in $A$ \st
  $f(A)$ ($A$ controlled by $f$) avoids \emph{Bad} and satisfies $\MeanCost(f(A)) \leq k$?
\end{prob}

\newcommand{\vtab}[1]{
  \begin{tabular}[c]{c}
    #1
  \end{tabular}
}
\begin{table}[t]
  \centering
  \caption{Summary of the Results}
  \label{tab-summary}
  \begin{tabular}[t]{||c|c|c|c||}\hline\hline
    &  Static Observers   & \multicolumn{2}{c||}{Dynamic Observers} \\\cline{3-4}
    &  Min. Cardinality   & Most Perm. Obs. & Optimal Observer \\\hline\hline
    DES & NP-Complete~\cite{cassez-acsd-07} & 2EXPTIME~\cite{cassez-acsd-07} & 2EXPTIME~\cite{cassez-tase-07}   \\\hline
    TA & PSPACE-Complete & 2EXPTIME  &  2EXPTIME-hard \\\hline\hline
  \end{tabular}
\end{table}

\bibliographystyle{IEEEtran}

\end{document}